%% file: appM2.tex
\documentclass[12pt]{article}
\usepackage{srcltx}
\usepackage{times,amsbsy,amsmath,amsthm,amssymb,amscd,amsfonts,mathrsfs}
\usepackage{graphicx,epstopdf}
\usepackage[hang]{subfigure}

\usepackage{color}

\newcommand\bOmega{\boldsymbol{\Omega}}
\newcommand\od[2]{\dfrac{\rd {#1}}{\rd {#2}}}
\newcommand\rd[0]{\mathrm{d}}
\newcommand{\pd}[2]{\dfrac{\partial #1}{\partial #2}}
\newcommand{\dd}{\,{\rm d}}
\newcommand\Beta{\boldsymbol{B}}
\newtheorem{remark}{Remark}[section]
\newtheorem{theorem}{Theorem}[section]

\newtheorem{example}[theorem]{Example}

\newcommand{\Vint}[1]{\left \langle #1 \right \rangle}
\newcommand{\vint}[1]{\langle #1 \rangle}

\def\bm{{\boldsymbol{m}}}
\def\bsalpha{\boldsymbol{\alpha}}

\def\bE{\boldsymbol{E}}
\def\bff{\boldsymbol{f}}
\def\br{\boldsymbol{r}}

\begin{document}

\title{An Approximate $M_2$ Model for Radiative Transfer in Slab
  Geometry}

\author{Graham Alldredge\thanks{Department of Mathematics RWTH Aachen
    University, Aachen, Germany. Email:
    \texttt{alldredge@mathcces.rwth-aachen.de}}, ~Ruo Li\thanks{CAPT,
    LMAM \& School of Mathematical Sciences, Peking University. Email:
    \texttt{rli@math.pku.edu.cn}} ~and~ Weiming Li\thanks{School of
    Mathematical Sciences, Peking University. Email:
    \texttt{liweimingsms@gmail.com}}}

\maketitle

\date{}

\input{abstract}

\input{introduction}

\input{preliminary}

\input{model}
\input{numerical_result}

\input{conclusion}

\section*{Acknowledgements}
The authors appreciate the financial supports provided by \emph{the
  National Natural Science Foundation of China (NSFC)} (Grant 91330205
and 11325102). We thank Mr. Kailiang Wu a lot for the discussion on 
the proof of the global hyperbolicity.

\bibliographystyle{plain}
\bibliography{appM2}

\end{document}

%% file: abstract.tex
\begin{abstract}
  We propose an approximate second order maximum entropy ($M_2$) model
  for radiative transfer in slab geometry. The model is based on the
  ansatz of the specific intensity in the form of a
  $\Beta$-distribution. This gives us an explicit form in its closure. The
  closure is very close to that of the maximum entropy, thus an
  approximation of the $M_2$ model. We prove that the new model is
  globally hyperbolic, sharing most of the advantages of the maximum
  entropy closure. Numerical examples illustrate that it provides
  solutions with satisfactory agreement with the $M_2$ model.
\end{abstract}

Keywords: Radiative transfer, slab geometry, maximum entropy, moment
model.

%% file: introduction.tex

\section{Introduction}
The radiative transfer equation describes the density of a system of 
particles interacting with a background medium.  It has been widely 
used in various applications such as atmospheric modeling, nuclear
engineering and medical imaging. As the radiative transfer equation is
a problem in very high dimension, deriving a low dimensional model is
the first step before further numerical studies. Most models for
radiative transfer, and more broadly, for kinetic equations in
general, fall in one of the following catogories: particle models,
moment models, and discrete-velocity models. In this paper we will
focus only on moment models.

Often, the moment models are equipped with the nice property of being
naturally rotational invariant, their variables have clear physical
meaning and therefore offer clear insight into the physics of the
problem under consideration. They can be highly efficient in many
applications, for example the Euler and Navier-Stokes equations in fluid
dynamics.
However, when creating moment model it can be difficult to ensure that
it is both hyperbolic and consistent with the fact that the unknown in
the radiative transfer equation is a density and thus must be
nonnegative.
This second property we call \emph{positivity} and is a major problem
with the $P_n$ method \cite{garrett2013comparison}, a linear method
which is one of the most well-known moment methods in the radiative
transfer community.
It was proven in
\cite{mcclarren2008solutions} that linear moment models that are also
hyperbolic and rotational invariant almost inevitably bring
non-positive solutions, thus motivating the study of nonlinear
models. However, it is not straightforward to give a hyperbolic
nonlinear model.%
\footnote{
For instance, the counterpart of $P_n$ model for the Boltzmann
equation is Grad's model \cite{grad1949kinetic}.
Unlike $P_n$, it is nonlinear, and is not globally hyperbolic.
}
A hyperbolic moment model successfully preserving positivity is the
maximum entropy model in radiative transfer.
It was first proposed by Minerbo
\cite{minerbo1978maximum}, and Levermore generalized it and exposed
its mathematical structure \cite{levermore1996moment}.
Unfortunately, the maximum entropy model
currently has no efficient implementation because the closure relation
is not explicit but must be computed through the solution of an
optimization problem. However, it is still
regarded as the most attractive model due to its highly desirable
mathematical properties.

While some recent work has been directed towards developing
efficient algorithms for $M_n$ \cite{a2001one, alldredge2014adaptive},
there has also been increasing effort devoted to its
approximation. In this work, we investigate a simple case, in
particular in slab geometry and where frequency dependence is omitted.
This includes the single-frequency and gray medium cases. 
When only the first three moments are specified, observations on the 
specific intensity maximizing the Bose-Einstein entropy in both cases gives
us the expectation that the specific intensity with maximal entropy can be
well approximated by a $\Beta$-distribution.
With such a form, an explicit expression of the
fourth moment as a function of the first three moments, is
obtained, resulting in a moment model with explicit closure.

For this new model, we illustrate that its closure is very similar to
the closure of $M_2$ model, both for single-frequency and gray case.
The three closures are consistent with an underlying nonnegative density
and agree exactly on the boundary of the
realizability region. In the interior of the realizability region, 
all three closures agree with each other qualitatively quite well.
Moreover, we show that the new model shares two important mathematical
properties of the $M_2$ model: global hyperbolicity and finite signal
speeds no larger than the speed of light.

The new approximate model is very convenient for numerical simulation
due to its explicit closure and conservative formulation. We present
several numerical results comparing with the results given by the
original $M_2$ model. Due to the serious difficulties in the
implementation of the $M_2$ model, the numerical results of the $M_2$
model is obtained by very complex techniques, precisely a revised
implementation based on \cite{alldredge2014adaptive}. Even so, it is
still very time consuming. The comparison of the numerical results are
quite satisfactory, considering the dramatic efficiency improvement by
the approximate model.

The rest of this paper is arranged as follows:
In Sec. \ref{sec:preliminary} we introduce the basics of moment models
and the maximum entropy closure.
In Sec. \ref{sec:approx} we propose an approximate
$M_2$ model and analyze its properties.
In Sec. \ref{sec:numerics} we
compare the approximate model with $M_2$ using several examples.
Finally in Sec. \ref{sec:conclude} we summarize and draw conclusions.

%% file: preliminary.tex

\section{Preliminaries}\label{sec:preliminary}
Let the specific intensity $I(t,\boldsymbol{x},\nu,\bOmega)$ to be
proportional to the density of radiation energy, which is a function
depended on the time $t\in\mathbb{R}^+$, the spatial coordinates
$\boldsymbol{x}\in\mathbb{R}^3$, the frequency $\nu\in\mathbb{R}^+$,
and the direction variable $\bOmega\in S^2$ on the unit sphere. It is
governed by the general form of radiative transfer equation as
\begin{equation}\label{eq:rt}
  \dfrac{1}{c}\pd{I}{t}+\bOmega\cdot\nabla I =\mathcal{C}(I),
\end{equation}
where $\mathcal{C}(I)$ describes the interaction of radiation with
background medium.

Denote by $\left\{m_i\left(\bOmega\right)\right\}$ a set of basis of a
polynomial space $\mathbb{H}_{\mathbf{N}} \left( \bOmega \right)$,
then the moments of the specific intensity $I$ are $\vint{\bm I}$
where we use the notation $\vint{\cdot}$ for either
$$
\Vint{g} := \int_{S^2} g(\Omega)\, dS(\Omega)
 \qquad \text{or} \qquad
\Vint{g} := \int_0^\infty \int_{S^2} g(\nu, \Omega)\, dS(\Omega) d\nu,
$$
where $dS(\Omega)$ is the volume element on the sphere.
The former leads only to an angular closure, while the latter leads to
the so-called gray approximations.
The exact moments $\vint{\bm I}$ satsify
\begin{equation}\label{eq:moment-eq}
    \dfrac{1}{c}\pd{\left\langle \boldsymbol{m}I\right\rangle}{t}
    +\nabla\cdot\left\langle\bOmega\boldsymbol{m}I\right\rangle
    = \left\langle\boldsymbol{m}\mathcal{C}(I)\right\rangle,
\end{equation}
where $\langle g \rangle$ can either mean $\int_{S^2} g\dd\bOmega$
or $\int_0^\infty\left(\int_{S^2} g\dd\bOmega\right)\dd\nu$.
However, the equation for the term
$\left\langle \bOmega\boldsymbol{m}I \right\rangle$ involves moments
of polynomials not in $\mathbb{H}_{\mathbf{N}} \left( \bOmega
\right)$, therefore \eqref{eq:moment-eq} is not a closed
system.
A moment model is then defined by approximating the higher-order moments
in terms of lower order moments to give a closed system of equations
approximating \label{eq:moment-eq} resulting in a system of equations
of the form
\begin{equation}\label{eq:moment-closure}
 \dfrac{1}{c}\pd{\bE}{t} + \nabla \cdot \bff(\bE) = \br(\bE),
\end{equation}
where $\bE \simeq \left\langle \boldsymbol{m}I\right\rangle$,
$\bff(\bE) \simeq \left\langle \bOmega\boldsymbol{m}I\right\rangle$, and
$\br(\bE) \simeq \left\langle\boldsymbol{m}\mathcal{C}(I)\right\rangle$.
How this closure is made is called a \emph{moment closure} and has
a fundamental impact on the performance of the moment model.

The maximum entropy principle is an elegant way of deriving moment closure.
It is based on reconstructing an ansatz of $I$ from the moments by 
solving the following constrained variational maximization problem 
\begin{equation}\label{eq:bose-entropy}
    \begin{aligned}
        &\max~~H(I)\\
        &s.t.\quad\langle I \boldsymbol{m}\rangle=\boldsymbol{E},\\
    \end{aligned}
\end{equation}
where $H(I)$ is the Bose-Einstein entropy
\begin{equation}\label{eq:be-entropy}
 H(I) := \left\langle \dfrac{k_B c^2}{2\hbar \nu^3} \left( -I \log(I)
  + \left(I + \dfrac{2\hbar\nu^3}{c^2} \right)
  \log\left(I + \dfrac{2\hbar\nu^3}{c^2} \right)
  \right)\right\rangle
\end{equation}
For the angular closure, the solution of \eqref{eq:bose-entropy} has
the form
\begin{equation}\label{eq:be-ansatz}
\hat{I}_{\bsalpha}(\boldsymbol{\Omega}) =
\dfrac{2\hbar\nu^3}{c^2} \left( \exp \left( \dfrac{\hbar \nu}{k_B
    }\bsalpha\cdot\boldsymbol{m} \right)-1 \right)^{-1},
\end{equation}
while for the gray approximations the solution of
\eqref{eq:bose-entropy} has the form
\begin{equation}\label{eq:gray-ansatz}
\hat{I}_{\bsalpha}(\boldsymbol{\Omega}) =
\frac{\sigma}{( \bsalpha \cdot \bm )^4},
\end{equation}
where $\sigma$ is the Stefan-Boltzmann constant.
In both cases $\bsalpha = \bsalpha(\bE)$ is the unique vector such
that $\vint{\hat{I}_{\bsalpha} \bm} = \bE$.
Then the $M_n$ method is defined by taking
\begin{equation}
 \bff(\bE) = \left\langle \bOmega\boldsymbol{m} \hat{I}_{\bsalpha}
  \right\rangle
 \qquad \text{and} \qquad
 \br(\bE) \simeq \left\langle\boldsymbol{m}\mathcal{C}
  (\hat{I}_{\bsalpha})\right\rangle
\end{equation}
in \eqref{eq:moment-closure}.

Properties of the $M_n$ model are discussed in
\cite{levermore1996moment,dubroca1999theoretical},
including a proof of its global hyperbolicity. It is also positivity
preserving, and entropy dissipating.  However, from
\eqref{eq:bose-entropy} one see that the closure is not given
explicitly. Instead one has to solve $\langle \hat{I}_{\bsalpha}
\boldsymbol{m} \rangle = \boldsymbol{E}$ for the Lagrange multipliers
$\bsalpha$, which involves solving a coupled nonlinear algebraic
system.
Unfortunately, it is
expensive and difficult to numerically solve this algebraic system. Due
to these numerical difficulties, there has so far been no efficient
general implementation of the $M_n$ model except in the
$M_1$ case \cite{berthon2007hllc, olbrant2012realizability}. However,
in some examples, the $M_1$ model is qualitatively wrong \cite{a2001one}.
Generally, there are two
approaches for resolving the difficulties in the implementation of the
maximum entropy model. One approach is to develop efficient algorithms
for solving the optimization problem. There has recently been some
progress in computing for high order $M_n$ models
\cite{alldredge2012high, alldredge2014adaptive}. The other approach is
to give an approximate model of $M_n$ which is explicit and
therefore more computationally feasible, while still preserving as
many of the advantages of the $M_n$ model as possible.

%% file: model.tex

\section{Approximate $M_2$ Model}
\label{sec:approx}
Due to the difficulties in deriving an approximate model, we
restrict ourselves to the radiative transfer equation in slab
geometry and consider only an approximation of $M_2$ model.
The radiative transfer equation becomes
\begin{equation}\label{eq:slab}
    \begin{array}{c@{\vspace{5pt}}c}
        \dfrac{1}{c}\pd{I(\mu)}{t}+\mu\pd{I(\mu)}{z}=
        -\sigma_a I(\mu)-\sigma_s\left(I(\mu)-\dfrac{I_0}{2}\right),\\
        \mu\in[-1,1],\quad I_0 = \int_{-1}^1 I(\mu)\dd\mu.\\
    \end{array}
\end{equation}
For the case where we only perform an angular closure, $I$, $\sigma_a$,
and $\sigma_s$ are still dependent on the frequency $\nu$, while in the
case of the gray approximations, we assume that $\nu$ has already been
integrated out of the equation.
Therefore from now on the angle-bracket notation indicates integrals
over $\mu$
$$
\Vint{g} = \int_{-1}^1 g(\mu) \, d\mu.
$$
Let $E_j \simeq \int_{-1}^1\mu^j I(\mu)\dd\mu$, and denote by
$\mathcal{M}$ the realizable moment vector space of $M_2$, then
\[
    \mathcal{M}=\left\{\left(E_0,E_1,E_2\right)\left|
        0 < E_2 < E_0, E_1^2 < E_0 E_2
    \right.\right\}.
\]
as shown in \cite{monreal2008higher}. We observed that the specific
intensity with maximal entropy may be well approximated by a
$\Beta$-distribution if we consider a $M_2$ model. This makes us take
the $\Beta$-distribution as an ansatz for $I$ :
\begin{equation}\label{eq:approxM2-distribution}
    \hat{I} = \dfrac{E_0}{2 \Beta(\xi,\eta)}\left(\dfrac{\mu+1}{2}
    \right)^{\xi-1}\left(\dfrac{1-\mu}{2}\right)^{\eta-1},\quad 
    \xi=\dfrac{\alpha}{\beta},\quad
    \eta = \dfrac{1-\alpha}{\beta}.
\end{equation}
We note that the combinations of $\Beta$-distribution was used as an
approximation for the specific intensity in \cite{vikas2013radiation},
though for a different purpose.
\begin{remark}
  The choice of $\Beta$-distribution as ansatz is somewhat arbitrary, but
  it has much flexibility, allowing for skewness and non-symmetry.  We
  will show later on that it captures the essential profile of the
  specific intensity.
\end{remark}

Self-consistency for the first to the third moment require
\[
    \alpha=\dfrac{E_1/E_0+1}{2},\quad 
    \beta=\dfrac{(E_1/E_0)^2-E_2/E_0}{E_2/E_0-1},
\]
which gives a non-linear closure as
\[
    E_3=\int_{-1}^1\mu^3 \hat{I}(\mu)\dd\mu
           =\dfrac{E_1(E_2^2+2E_1^2-3E_0 E_2)}{2E_1^2-E_0 E_2-E_0^2}.
\]

On the boundaries of $\mathcal{M}$, since there is only one
nonnegative ansatz with the correct moments \cite{CurFial91} and the
$\Beta$ distribution is clearly nonnegative, our closure agrees with
the $M_2$ closure.
Indeed,
\begin{enumerate}
\item If $E_2/E_0=1$, then the specific intensity ansatz of $M_2$ is
  $I = \dfrac{1}{2} \left( E_0+ E_1 \right) \delta(\mu-1) +
  \dfrac{1}{2} \left( E_0 - E_1 \right) \delta(\mu+1)$, for which
  $\Beta$-closure shares with $M_2$ the same closure, and $E_3=E_1$.
\item If $E_2/E_0=(E_1/E_0)^2$, then the specific intensity of $M_2$
  is $I=E_0\delta(\mu-E_1/E_0)$, for which $\Beta$-closure also shares
  the same closure with $M_2$, thus $E_3=E^3_1/E^2_0$.
\end{enumerate}
Furthermore, this closure is correct in the isotropic case (that is, 
when $\bE = (E_0, 0, E_0/3)$ contains the moments of the constant
density $I(\mu) \equiv E_0 / 2$).

In Figure \ref{fig:compare-contour} we plot the contours of $E_3/E_0$
on $\mathcal{M}$ for a comparison
between the $M_2$ model and the $\Beta$-closure model.
Clearly the models agree qualitatively quite well.
\begin{figure} 
  \subfigure{  
    \includegraphics[width=0.3\textwidth]{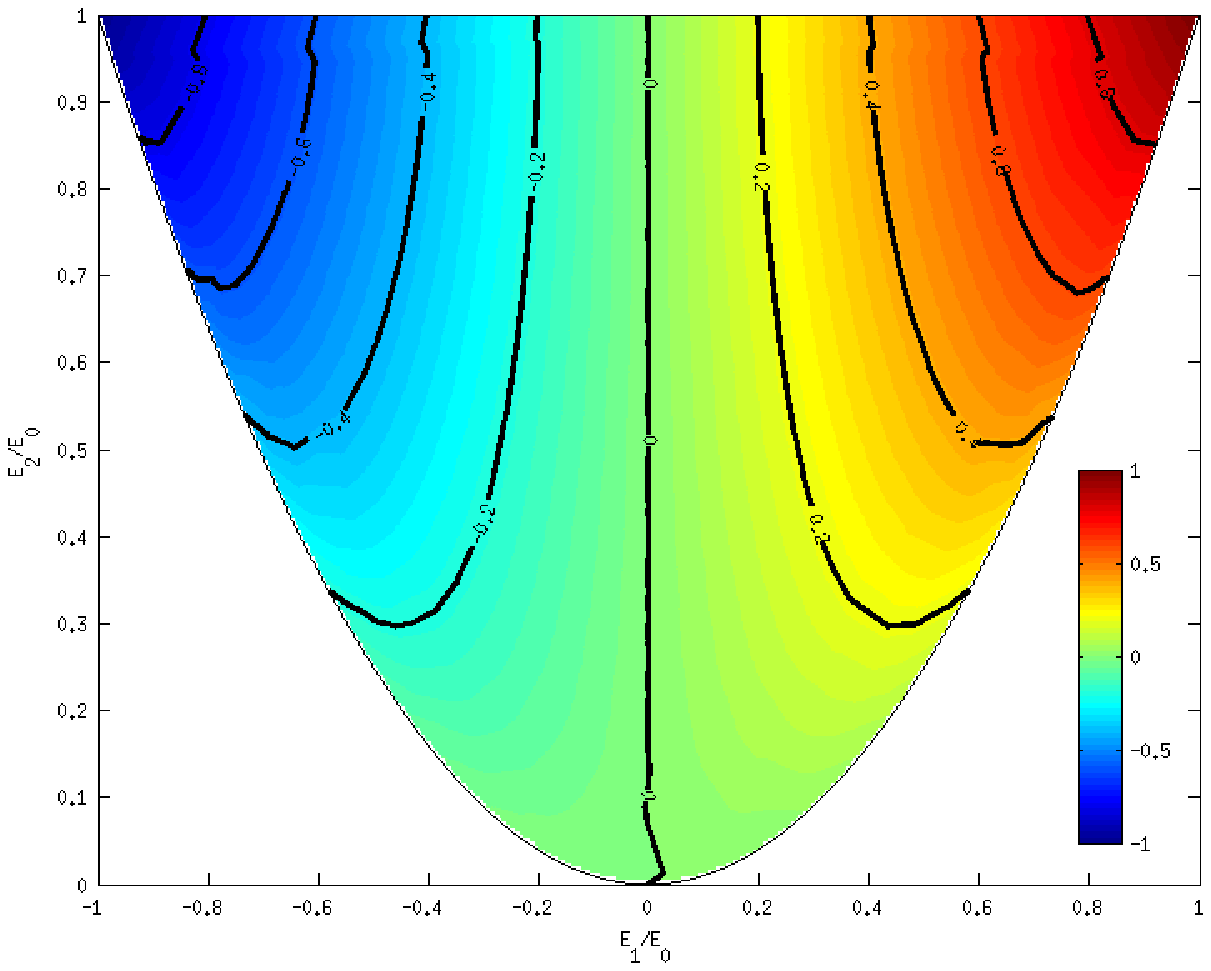}} 
  \subfigure{  
    \includegraphics[width=0.3\textwidth]{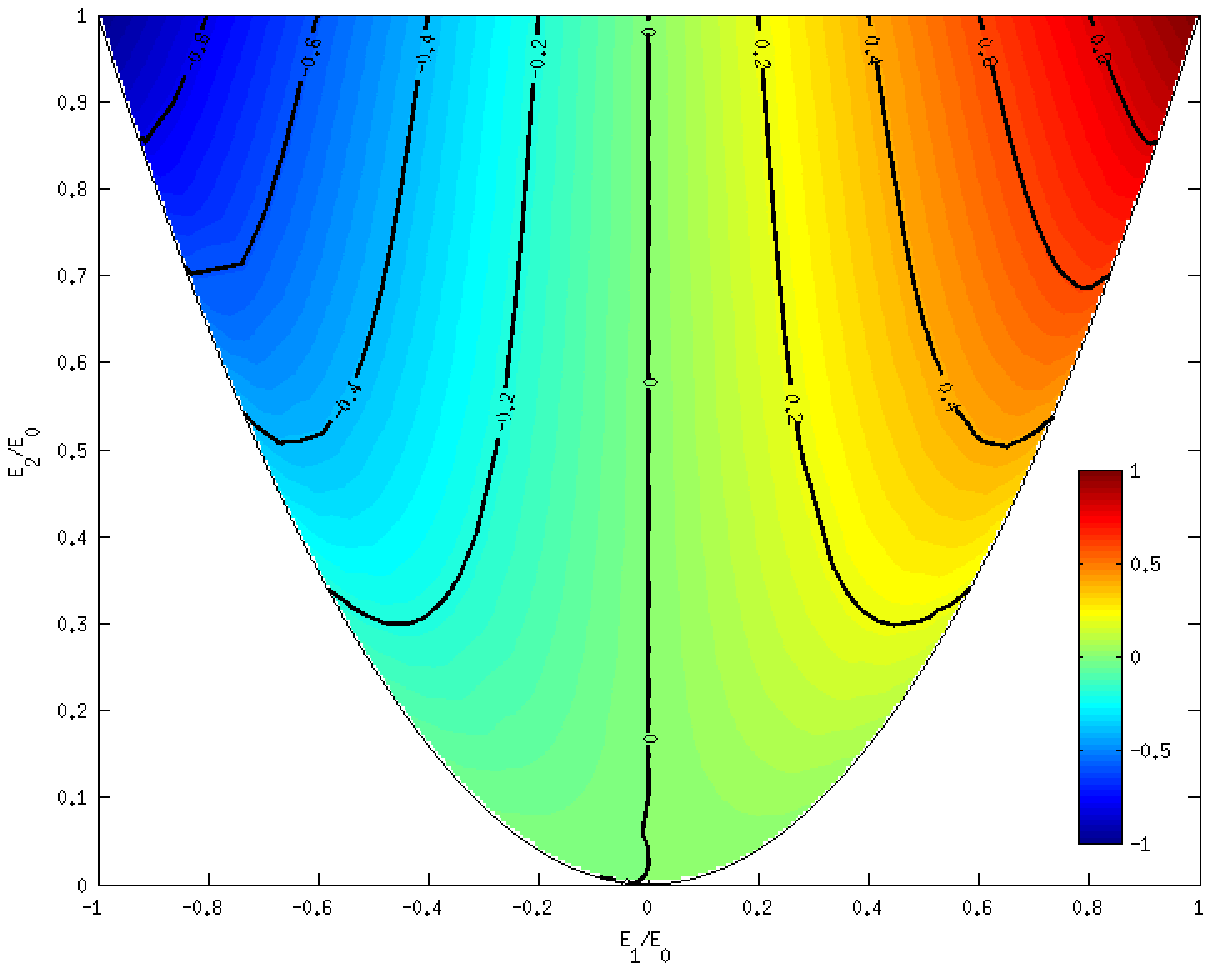}} 
  \subfigure{
    \label{fig:beta-contour} 
    \includegraphics[width=0.3\textwidth]{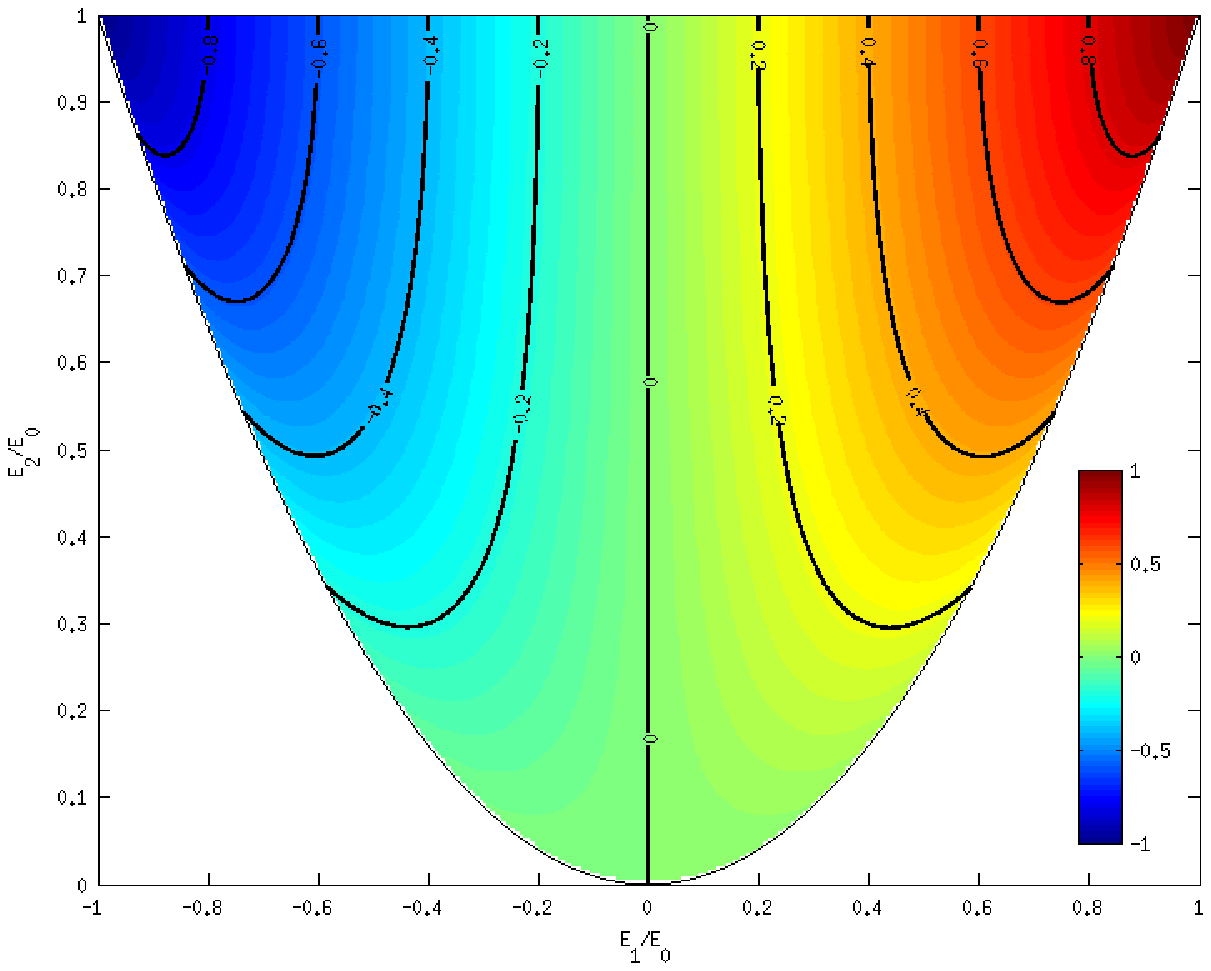}}  
  \caption{Comparing contour between $M_2$ for single frequency (left),
  $M_2$ for gray approximation (center), and $\Beta$-closure (right)} 
  \label{fig:compare-contour} 
\end{figure}

For four sets of moments in typical regions, Figure
\ref{fig:compare-distribution} compares the specific intensity between
the $M_2$ model and our $\Beta$-closure model. They all
qualitatively agree with each other.
\begin{figure} 
  \subfigure{  
    \label{fig:centerlower} 
    \includegraphics[width=0.48\textwidth]{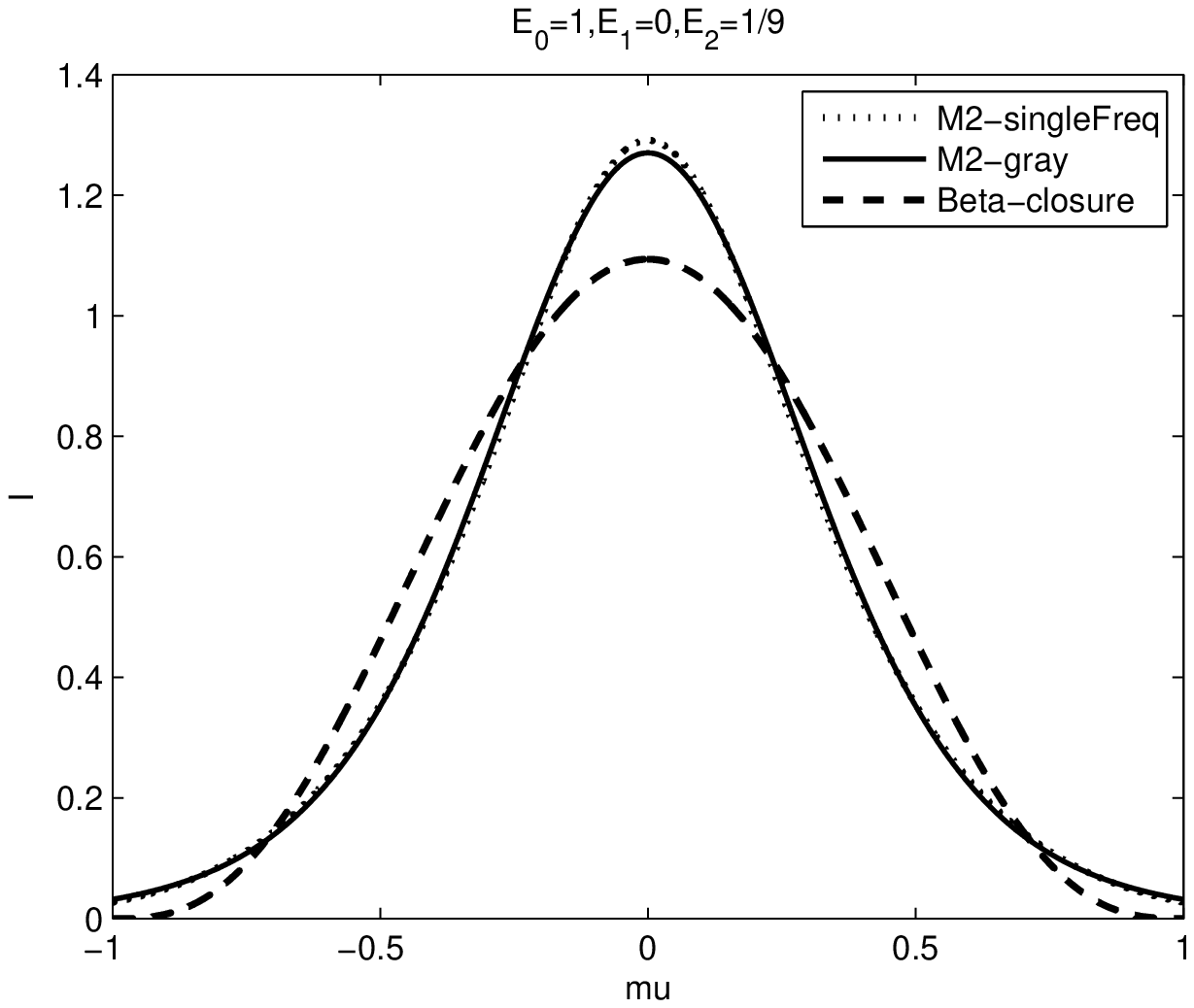}} 
  \hfill
  \subfigure{
    \label{fig:rightboundary}
    \includegraphics[width=0.48\textwidth]{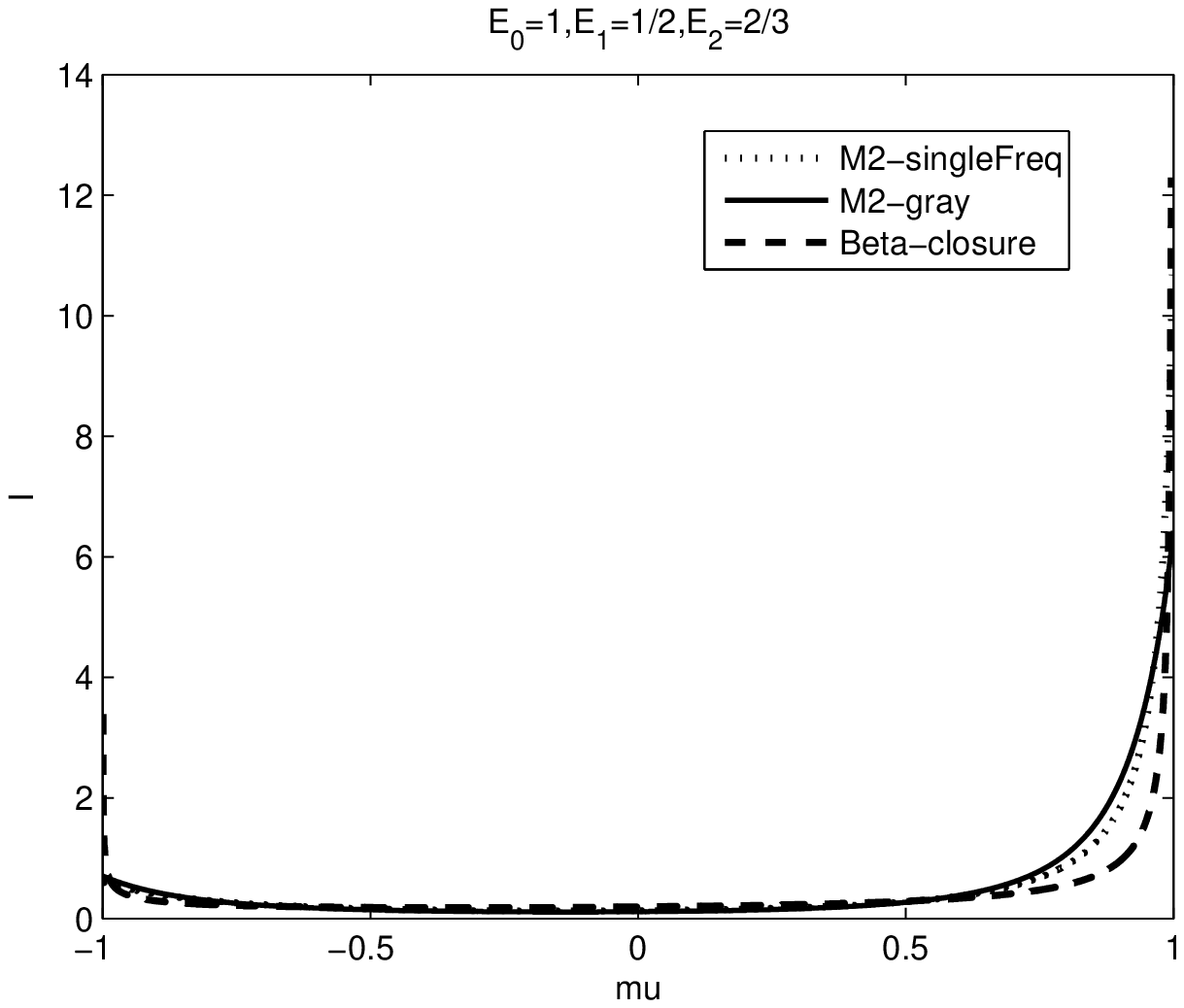}}
  \hfill
  \subfigure{
    \label{fig:rightcorner}
    \includegraphics[width=0.48\textwidth]{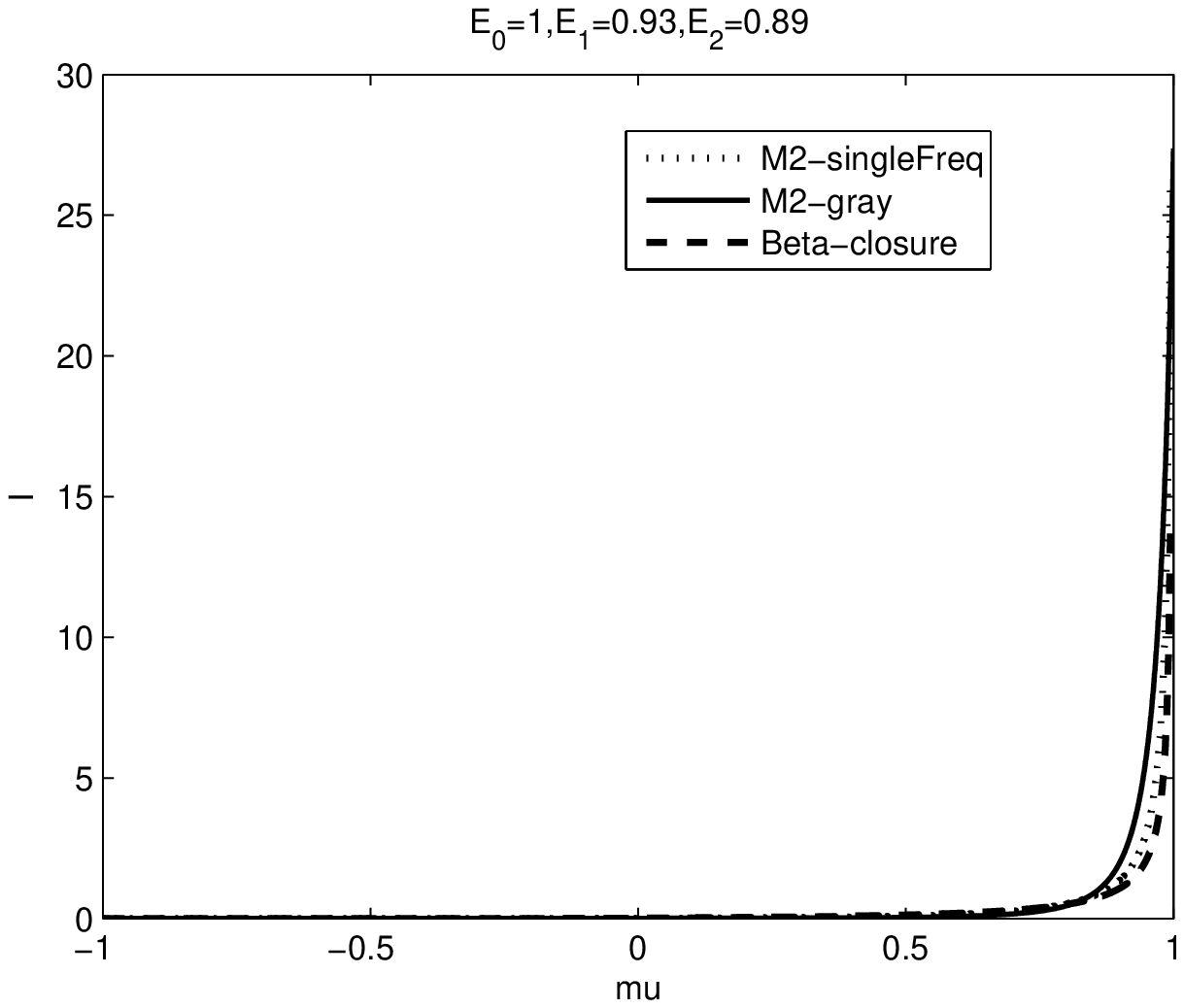}}
  \hfill
  \subfigure{
    \label{fig:rightcenter}
    \includegraphics[width=0.48\textwidth]{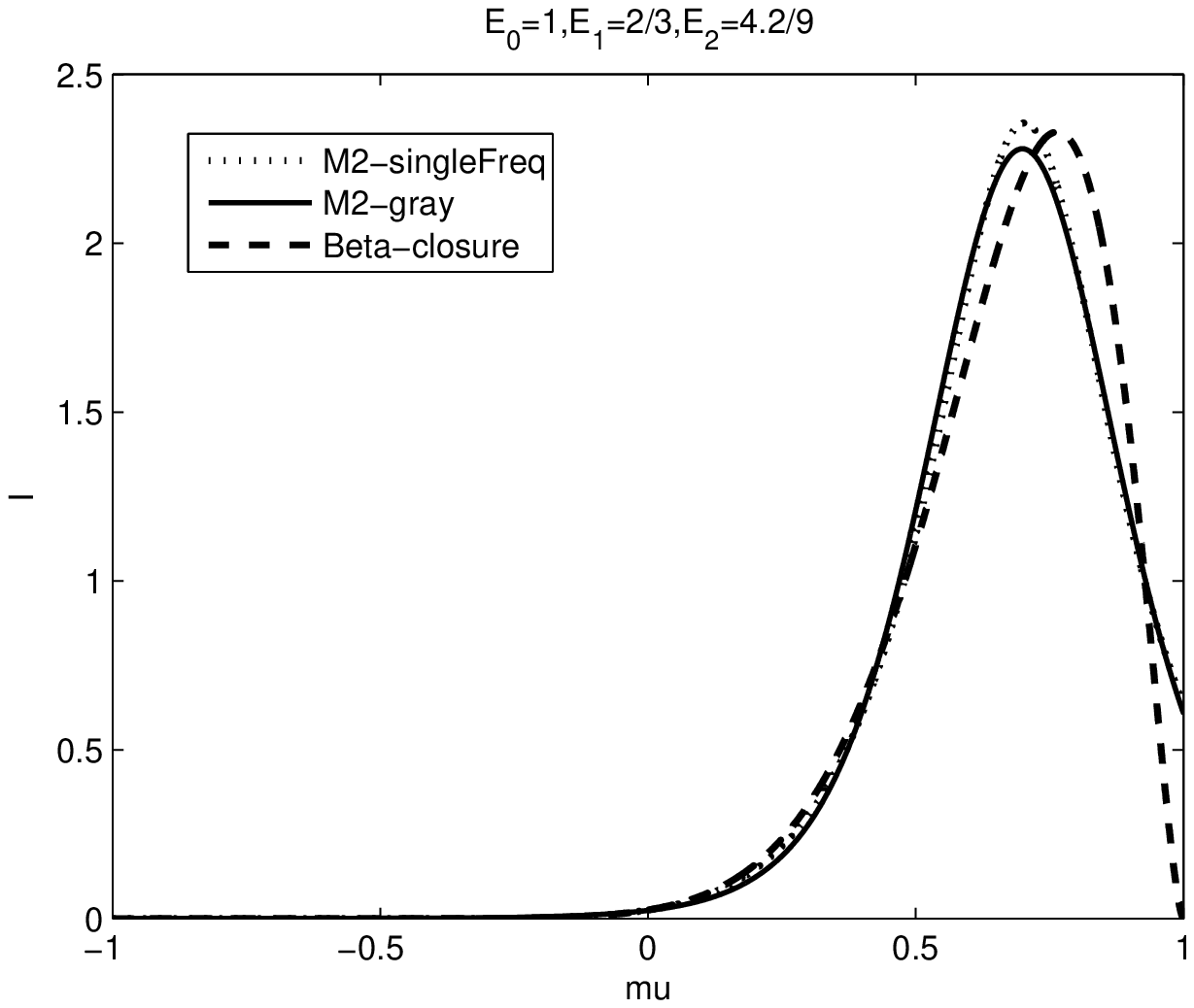}}
  \caption{Comparing specific intensity between $M_2$ and $\Beta$-closure model} 
  \label{fig:compare-distribution} 
\end{figure}

Let us show firstly that the new model based on the
$\Beta$-distribution is globally hyperbolic. The Jacobian matrix of
the approximate model is
\begin{equation}\label{eq:Jacobi-of-ApproxM2}
    \boldsymbol{J}=\left(\begin{array}{ccc}
            0 & 1 & 0\\
            0 & 0 & 1\\
            \pd{E_3}{E_0} & \pd{E_3}{E_1}& \pd{E_3}{E_2}\\
        \end{array}
    \right).
\end{equation}
Let $a_0=\pd{E_3}{E_0}$, $a_1=\pd{E_3}{E_1}$, and $a_2=\pd{E_3}{E_2}$,
they satisfy
\[
\begin{split}
  a_0&=-{\dfrac {E_{{1}} \left( E_{{0}}-E_{{2}} \right)  \left( 3\,E_{{0}}E_{{
            2}}-4\,{E_{{1}}}^{2}+{E_{{2}}}^{2} \right) }{ \left( {E_{{0}}}^{2}+E_{
          {0}}E_{{2}}-2\,{E_{{1}}}^{2} \right) ^{2}}},\\
  a_1 &={\frac { \left( 3\,{E_{{0}}}^{2}-E_{{0}}E_{{2}}-2\,{E_{{1}}}^{2}
      \right)  \left( E_{{0}}E_{{2}}-2\,{E_{{1}}}^{2}+{E_{{2}}}^{2}
      \right) }{ \left( {E_{{0}}}^{2}+E_{{0}}E_{{2}}-2\,{E_{{1}}}^{2}
      \right) ^{2}}},\\
  a_2 &= {\frac {E_{{1}} \left( E_{{0}}-E_{{2}} \right)  \left( 3\,{E_{{0}}}^{2
        }+E_{{0}}E_{{2}}-4\,{E_{{1}}}^{2} \right) }{ \left( {E_{{0}}}^{2}+E_{{0
          }}E_{{2}}-2\,{E_{{1}}}^{2} \right) ^{2}}}.
\end{split}
\]
The characteristic polynomial of $\boldsymbol{J}$ is 
\begin{equation}\label{eq:characteristic}
        p(\lambda)=\lambda^3-a_2\lambda^2-a_1\lambda-a_0.
\end{equation}
We then have the following theorem:
\begin{theorem}[Global strict hyperbolicity of $\Beta$-closure model]
  The $\Beta$-closure model is globally strictly hyperbolic in the
  interior of the realizable region $\mathcal{M}$, and its propagation
  speed is less than the speed of light.
\end{theorem}

\begin{proof}
  Let us study it in two cases:
    \begin{enumerate} 
       \item For $E_1=0$,
            \[
                p(\lambda)=\lambda\left(\lambda^2-
                \dfrac{E_2(3E_0-E_2)}{E_0(E_0+E_2)}\right).
            \]
            Clearly, the three distinct roots of $p(\lambda)$ are
            $\lambda_0 = 0$ and $\lambda_\pm = \pm \sqrt {\dfrac {E_2
                (3E_0-E_2)} {E_0(E_0+E_2)}}$. Since
            $0<\dfrac{E_2(3E_0-E_2)}{E_0(E_0+E_2)}<1$, all the roots
            are within $(-1,1)$.

        \item Consider $E_1\not=0$. Without loss of generality, assume $E_1>0$.
            Then
            \[
                -1 < 0 < \dfrac{a_2}{3} < \dfrac{E_1}{E_0} < 1.
            \]
            Let $Q=\dfrac{E_2}{E_0}-\left(\dfrac{E_1}{E_0}\right)^2$,
            $Z=1-\dfrac{E_2}{E_0}$, then
            \[
                p\left(\dfrac{a_2}{3}\right) = \dfrac{4}{27}\dfrac{E_1}{E_0}
                Q^2 Z f(Q,Z) > 0,
            \]
            with
            \[
                f(Q,Z) = \dfrac{68Q^2 Z^2+86Q Z^3+27Z^4+288 Q^3+360 Q^2 Z+112 Q Z^2}
                {(2Q+Z)^6}.
            \]
            We notice that
            \[
                \begin{split}
                    & p(-1) = -\dfrac{(E_0-E_2)^2(E_0+E_1)(E_0+2E_1+E_2)}
                    {(E_0^2+E_0E_2-2E_1^2)^2} < 0,\\
                    & p\left(\dfrac{E_1}{E_0}\right) = -\dfrac{4E_1}{E_0^3}
                    \dfrac{(E_0-E_1)(E_0+E_1)(E_0 E_2-E_1^2)^2}
                    {(E_0^2+E_0 E_2-2E_1^2)^2} < 0,\\
                    & p(1) = \dfrac{(E_0-E_2)^2(E_0-E_1)(E_0-2E_1+E_2)}
                    {(E_0^2+E_0 E_2-2E_1^2)^2} > 0,\\
                \end{split}
            \]
            thus $p(\lambda)$ has one root in each intervals $\left( -1,
              \dfrac{a_2}{3}\right)$, $\left( \dfrac{a_2}{3},
              \dfrac{E_1}{E_0} \right)$, and $\left( \dfrac{E_1}{E_0},
              1 \right)$. Similar arguments work for $E_1 <0$.
    \end{enumerate}
    This ends the proof.
\end{proof}

Denote $\lambda_1<\lambda_2<\lambda_3$ to be the three eigenvalues of
the Jacobian matrix $\boldsymbol{J}$, then it is clear that the
eigenvector corresponding to $\lambda_j$ is
$\boldsymbol{R}^{(j)}=\left[1~~ \lambda_j~~ \lambda_j^2\right]^T$.

\begin{theorem}\label{theorem:wave}
    The $\boldsymbol{R}^{(1)}$, $\boldsymbol{R}^{(3)}$-characteristic 
    fields are genuinely non-linear, while the $\boldsymbol{R}^{(2)}$
    -characteristic field is neither genuinely non-linear nor linearly
    degenerate.\footnote{For the definition of genuiely non-linear and
    linearly degenerate characteristic fields, see \cite{toro2009riemann}.}
\end{theorem}

\begin{proof}
    Let  
    \[
        \Delta(\lambda)=
        \left[1~~ \lambda~~ \lambda^2\right]\cdot\dfrac{\partial(a_0,a_1,a_2)}
        {\partial(E_0,E_1,E_2)}\cdot 
        \left[
            \begin{array}{c}
                1\\
                \lambda\\
                \lambda^2\\
            \end{array}
        \right].
    \]
    As
    \[
    \nabla\lambda_j\cdot\boldsymbol{R}^{(j)} = \left(
      \left.\od{p}{\lambda}\right|_{\lambda=\lambda_j} \right)^{-1}
    \Delta(\lambda),
    \]
    and             
    \[
        \left.\od{p}{\lambda}\right|_{\lambda=\lambda_j}>0,j=1,3;\quad 
        \left.\od{p}{\lambda}\right|_{\lambda=\lambda_2}<0,
    \]
    $\nabla\lambda_j\cdot\boldsymbol{R}^{(j)}=0$ is equivalent to
    $\lambda_j$ being the common root of $\Delta(\lambda)$ and
    $p(\lambda)$.  As the resultant of $p(\lambda)$ and
    $\Delta(\lambda)$ is
    \[
        \mathrm{res}(p,\Delta,\lambda)= 128\,{\frac {E_{{1}}{Q}^{6}{Z}^{4} \left( {Z}^{2}+16\,Q+8\,Z \right) 
            \left( Q+Z \right)  \left( {Z}^{2}+4\,Q \right) ^{3}}{ E^4_0 \left( 2\,Q+Z
        \right) ^{15}}},
    \]
    where $Q=\dfrac{E_2}{E_0}-\left(\dfrac{E_1}{E_0}\right)^2$,
    $Z=1-\dfrac{E_2}{E_0}$.  As $Q$ and $Z$ are both positive in the
    interior of the realizable region $\mathcal{M}$, we have that
    $\mathrm{res}(p,\Delta,\lambda)\neq 0$ if $E_1\neq 0$. Therefore
    when $E_1\neq 0$, all characteristic fields satisfy
    $\nabla\lambda_j\cdot\boldsymbol{R}^{(j)}\neq 0$.

    In case that $E_1=0$, $\lambda_2=0$ is a common root of
    $p(\lambda)$ and $\Delta(\lambda)$, so
    \[
        \nabla\lambda_2\cdot\boldsymbol{R}^{(2)}=0.
    \]
    Meanwhile we have
    \[
        \Delta(\lambda)=2\,{\lambda \left( {\lambda}^{2}E_{{0}}-E_{{2}} \right) \frac { \left( 3\,E_{{0}}+
            E_{{2}} \right)  \left( E_{{0}}-E_{{2}} \right) }{{E_{{0}}}^{2}
        \left( E_{{0}}+E_{{2}} \right) ^{2}}},
    \]
    thus $\nabla\lambda_j\cdot\boldsymbol{R}^{(j)}\not=0$ for $j=1,3$.

    Collecting the arguments above, one has that
    \[
        \begin{split}
            &\nabla\lambda_j\cdot\boldsymbol{R}^{(j)}\not=0,\quad\forall (E_0,E_1,
            E_2)\in\mathcal{M},\quad\text{for}~~j=1,3;\\
            &\mathrm{sign}\left(\nabla\lambda_j\cdot\boldsymbol{R}^{(2)}\right)
            =\mathrm{sign}(E_1),\quad\text{under proper scaling of~} \boldsymbol{R}^{(2)}.
        \end{split}
    \]
\end{proof}

We point out that it is also valid for the $M_2$ model that
$\nabla\lambda_2\cdot\boldsymbol{R}^{(2)}=0$ when $E_1=0$. Actually,
the specific intensity of the $M_2$ model is
\[
    I(\mu)=\dfrac{\alpha_0}{(1+\alpha_1\mu+\alpha_2\mu^2)^4},
\]
for the gray case, and
\[
    I(\mu)=\dfrac{\alpha_0}{\exp(1+\alpha_1\mu+\alpha_2\mu^2)-1},
\]
for the single-frequency case.
In this formation, the denominator is always positive on $\mu\in[-1,1]$, and we have that
$E_1=0$ implies $\alpha_1=0$. Therefore,
\[
E_3=0,~~\lambda_2 = 0,\quad\forall~(E_0,0,E_2)\in\mathcal{M}.
\]
Direct calculations show $\nabla\lambda_2 \cdot\boldsymbol{R}^{(2)}\left|_
{(E_0,0,E_2)}\right.=\pd{a_0}{E_0}\left|_{(E_0,0,E_2)}\right.=0$.

Let us summarize briefly some advantages of the new model:
\begin{enumerate}
\item The ansatz for the specific intensity preserves positivity and
  has explicit closure relationship;
\item The model derived is conservative and globally hyperbolic; 
\item The signal speed of the new model is less than the speed of
  light;
\item The closure is very close to that of $M_2$;
\end{enumerate}
We present some numerical results for some benchmark problems to show
the quality of the new model as an approximation of the $M_2$ model.

%% file: numerical_result.tex

\section{Numerical Results}\label{sec:numerics}
Our approximate $M_2$ model of \eqref{eq:slab} is therefore
\begin{equation}\label{eq:approxM2-moment}
    \begin{aligned}
        &\dfrac{1}{c}\pd{E_0}{t}+\pd{E_1}{z}=-\sigma_a E_0,\\
        &\dfrac{1}{c}\pd{E_1}{t}+\pd{E_2}{z}=-\left[\sigma_a+\sigma_s\right] E_1,\\
        &\dfrac{1}{c}\pd{E_2}{t}+\pd{E_3}{z}=-\sigma_a E_2
        +\sigma_s\left(\dfrac{E_0}{3}-E_2\right).\\
    \end{aligned}
\end{equation}
We consider the angular closure for a single frequency $\nu = 1$.
We solve equation \eqref{eq:approxM2-moment} using the canonical
finite volume scheme with the Lax-Friedrich numerical flux, and the
source term is treated implicitly.

To impose an inflow boundary condition, we only need to impose the value
of the flux on the boundary.  We derive it using upwind on the kinetic
scale. As we know the moments on the left and right cells, we can
reconstruct $\hat{I}(E^l_0, E^l_1, E^l_2, \mu)$ and $\hat{I}(E^r_0,
E^r_1, E^r_2,\mu)$ on the left and right cells using the ansatz
\eqref{eq:approxM2-distribution}, then integrate over $\mu$ to have
\[
    E_j = \int_0^1 \mu^j \hat{I}(\mu,E^l_0, E^l_1,E^l_2)\dd\mu + 
    \int_{-1}^0 \mu^j \hat{I}(\mu,E^r_0, E^r_1, E^r_2)\dd\mu,\quad\text{for} ~~j=1,2,3,
\]
to give the flux on the boundary.

We compare our model to numerical solutions of the true $M_2$ model,
which for the single-frequency case uses the ansatz \eqref{eq:be-ansatz}.
We compute $M_2$ solutions using the kinetic scheme and optimization
techniques given in \cite{alldredge2014adaptive}.
The entropy from that work is replaced with the Bose-Einstein entropy
\eqref{eq:be-entropy}, and to avoid the singularity in the ansatz
\eqref{eq:be-ansatz} when the polynomial $\boldsymbol{\alpha} \cdot
\boldsymbol{m}$ passes through zero, we limit the step-size in the
Armijo line search so that the polynomial $\boldsymbol{\alpha} \cdot
\boldsymbol{m}$ remains negative at every angular quadrature point.
Our solutions are computed with 1000 cells, and we note that none of
the computations below required the use of the isotropic
regularization technique.

Below we give the numerical results for three examples.
\begin{example}[Two-beam]
  The
  absorption coefficient is $\sigma_a = 2$, the scattering coefficient
  $\sigma_s = 0$, and the speed of light is taken to be $c=1$. The
  spatial domain is $z\in[0,1]$. The initial value is set as
  $E_0=10^{-7}$, $E_1=0$, $E_2=\dfrac{1}{3} E_0$ for all $z$. Inflow
  boundary condition are imposed on both ends, thus the specific
  intensity on the boundaries are $I(t,0,\mu)= \exp
  (-10(\mu-1)^2)$ on the left boundary $z=0$, and $I(t,1,\mu)=\exp
  (-10(\mu+1)^2)$ on the right boundary $z=1$.
\end{example}
In Figure \ref{fig:2beam} we compares the steady-state solution of
$E_0$ between the $M_2$ model and the $\Beta$-closure model.
\begin{figure}
    \centering
    \includegraphics[width=0.7\textwidth]{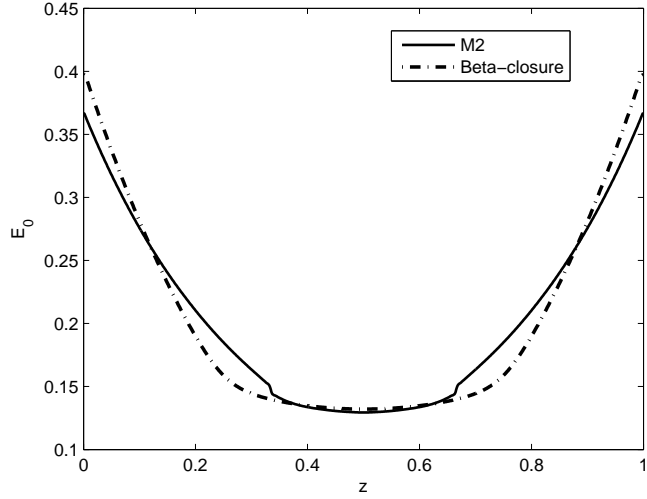}
    \caption{steady state solution for two-beam problem} 
\label{fig:2beam}
\end{figure}

\begin{example}[Isotropic inflow into vacuum]
  In this example, we consider the spatial domain $z\in[-\infty,1]$
  with an isotropic inflow source is imposed on the right
  boundary into a domain which is unbounded on the left.
  We take $\sigma_a = \sigma_s = 0$.
  Initially, for all
  $z$ , we take $E_0=10^{-8}$, $E_1=0$, $E_2=\dfrac{1}{3}E_0$. The
  isotropic inflow is specified at $z=1$. The specific intensity
  outside the right boundary is $I(\mu)=0.5$. We carry out the
  computation from $t_0=0$ to $t= 0.5418$ and $0.8$.
\end{example}
The results are in Figure \ref{fig:inflow}, which are the value of
$E_0$ at $t=0.5418$ and $t=0.8$ for both the $\Beta$-closure model and
the $M_2$ model.
\begin{figure}
    \centering
    \subfigure{  
        \label{fig:a1} 
    \includegraphics[width=0.48\textwidth]{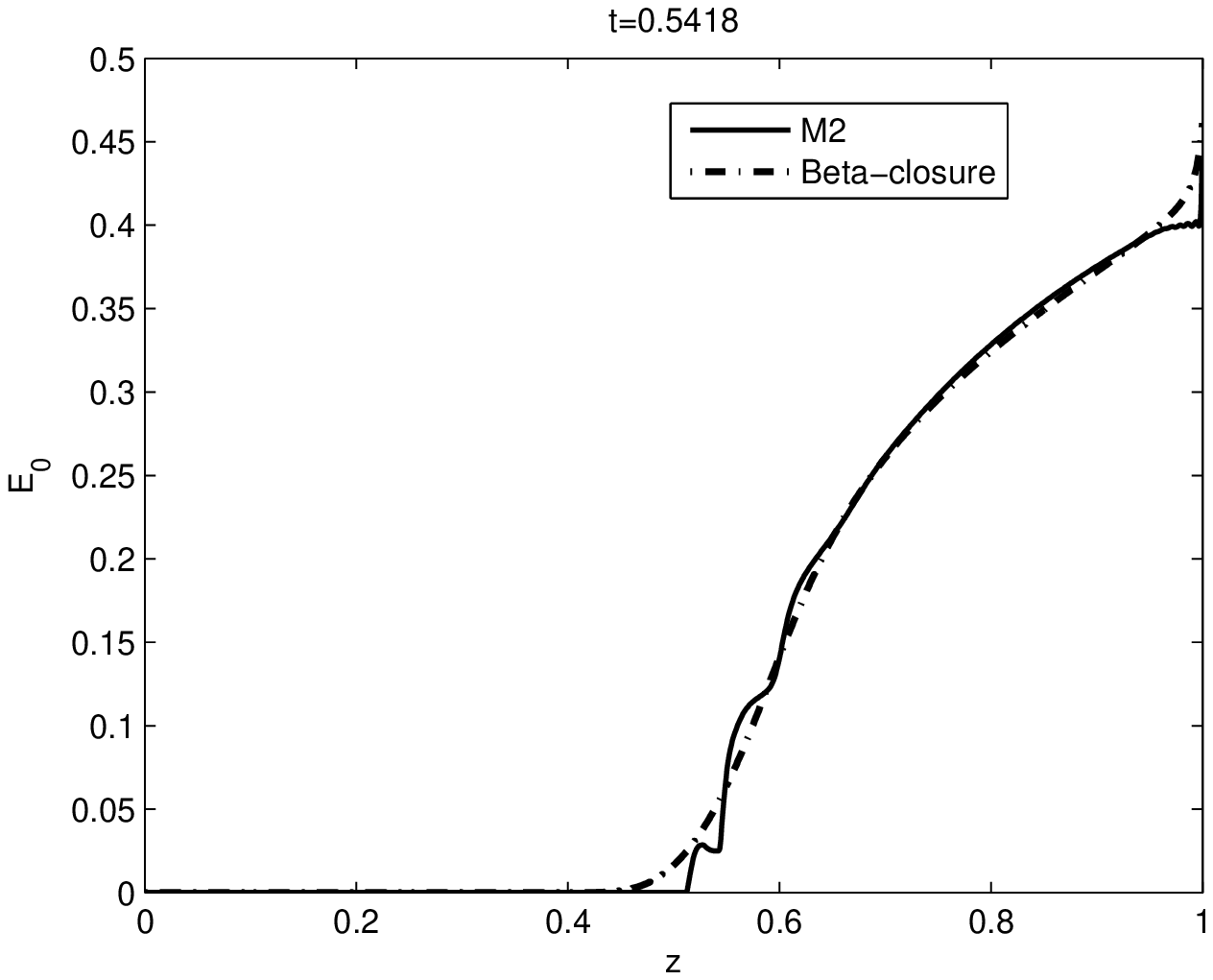}} 
    \subfigure{
        \label{fig:b1} 
    \includegraphics[width=0.48\textwidth]{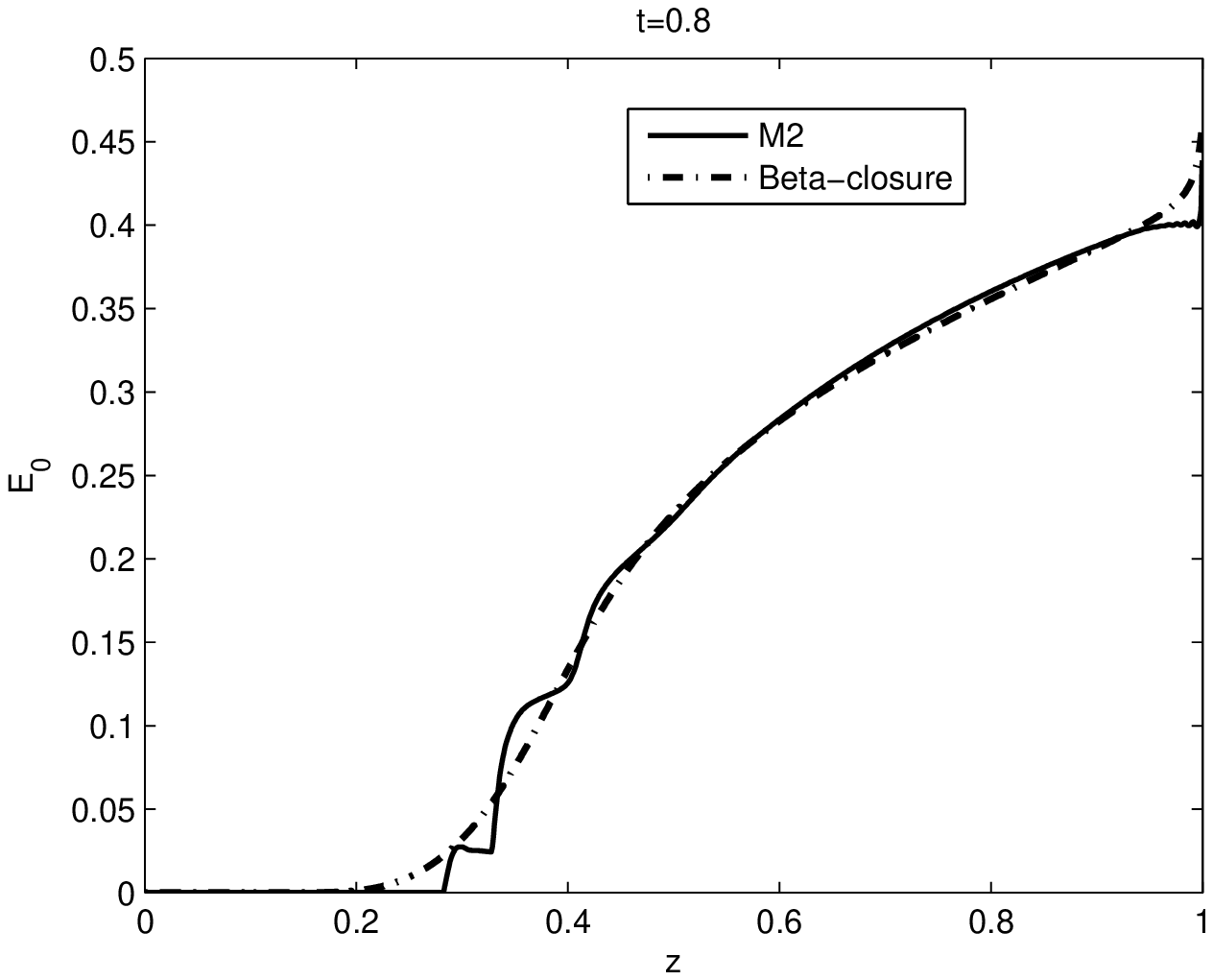}} 
    \caption{Results for inflow into vacumm problem} 
\label{fig:inflow}
\end{figure}

\begin{example}[Plane source]
  In this test the spatial domain is unbounded, and the initial
  value is taken as $E_0(z)=\delta(z)+ 10^{-8}$, $E_1=0$ and
  $E_2=\dfrac{1}{3}E_0$. The simulation time interval is from $t_0=0$
  to $t=0.7005$ and $t=0.9043$.
\end{example}
The numerical results of $E_0$ for $M_2$ and the $\Beta$-closure
model are in Figure \ref{fig:plane-source}.
\begin{figure}
    \centering
    \subfigure{  
        \label{fig:a2} 
    \includegraphics[width=0.48\textwidth]{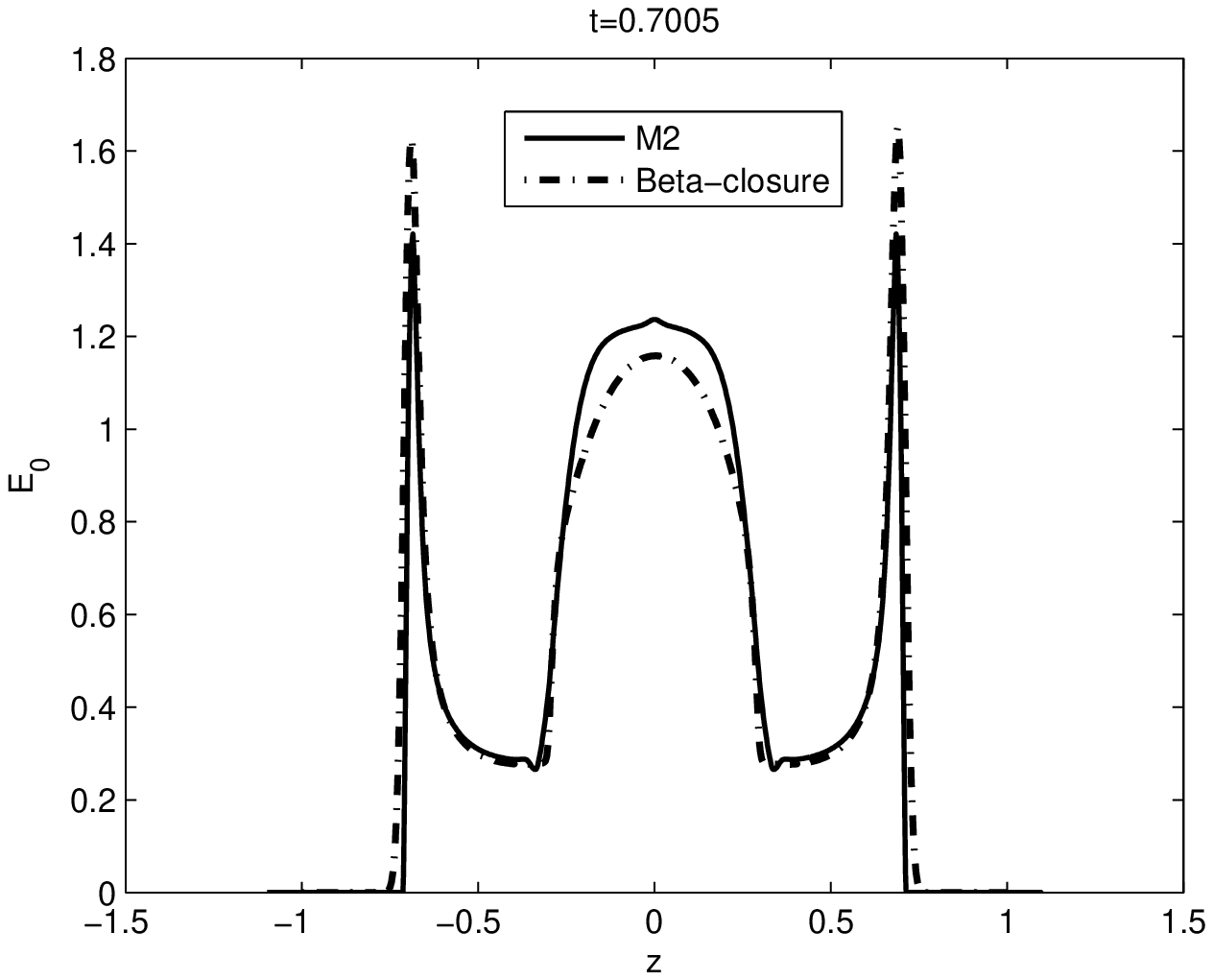}} 
    \subfigure{
        \label{fig:b2} 
    \includegraphics[width=0.48\textwidth]{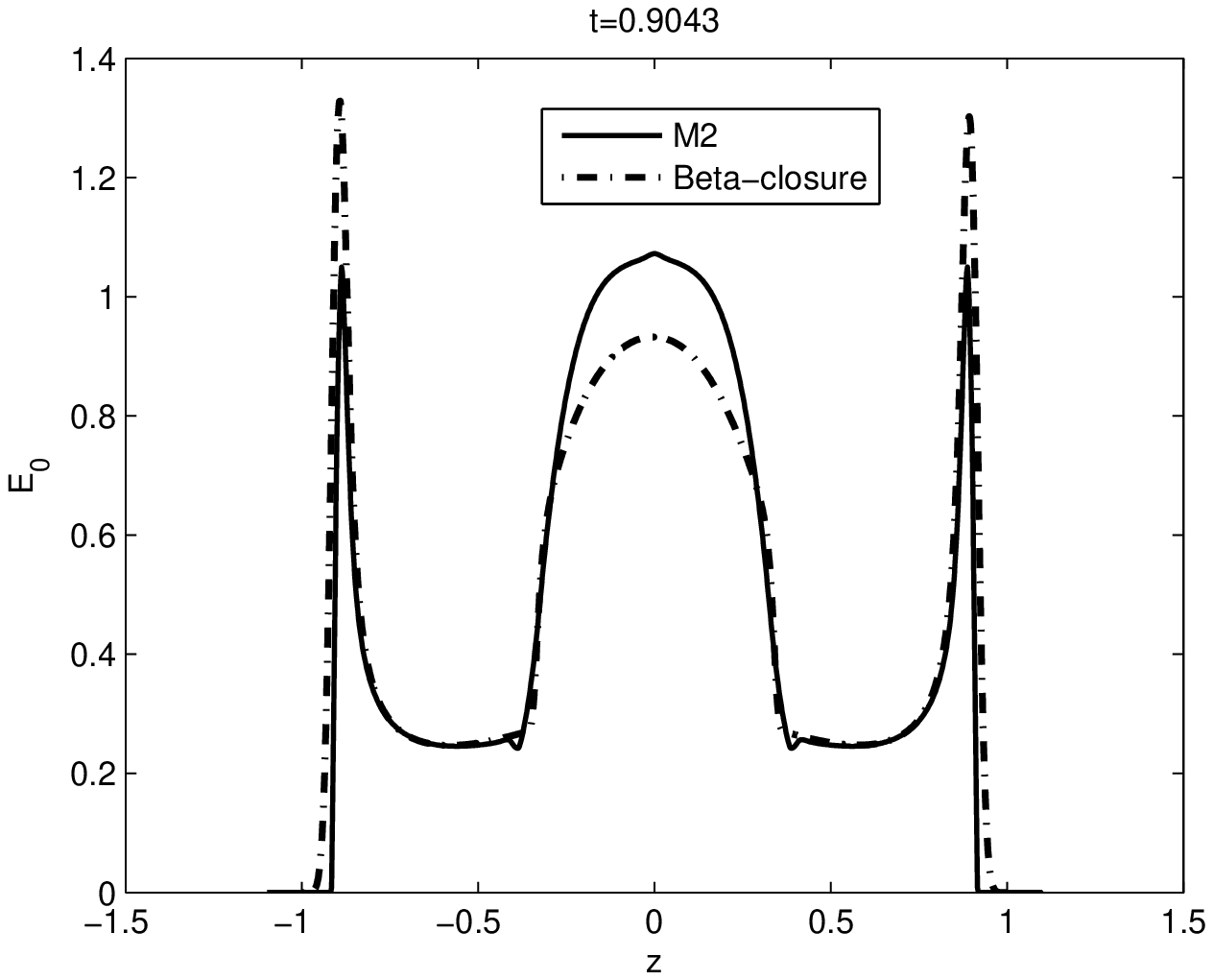}} 
    \caption{Results for plane source problem} 
\label{fig:plane-source}
\end{figure}

%% file: conclusion.tex
\section{Conclusion}\label{sec:conclude}
An approximate $M_2$ model for the radiative transfer 
slab geometry in the cases of single-frequency and grey medium
is proposed. The new model is based on an
ansatz formulated as a $\Beta$-distribution. It shares most of the
advantages of the $M_2$ model while it has an explicit closure. We are
now working the extension of this idea to a three dimensional
configuration and a many moment model.

%% file: appM2.bbl
\begin{thebibliography}{10}

\bibitem{a2001one}
Thomas A~Brunner and James Paul~Holloway.
\newblock One-dimensional riemann solvers and the maximum entropy closure.
\newblock {\em Journal of Quantitative Spectroscopy and Radiative Transfer},
  69(5):543--566, 2001.

\bibitem{alldredge2014adaptive}
Graham~W Alldredge, Cory~D Hauck, Dianne~P OʼLeary, and Andr{\'e}~L Tits.
\newblock Adaptive change of basis in entropy-based moment closures for linear
  kinetic equations.
\newblock {\em Journal of Computational Physics}, 258:489--508, 2014.

\bibitem{alldredge2012high}
Graham~W Alldredge, Cory~D Hauck, and Andr{\'e}~L Tits.
\newblock High-order entropy-based closures for linear transport in slab
  geometry ii: A computational study of the optimization problem.
\newblock {\em SIAM Journal on Scientific Computing}, 34(4):B361--B391, 2012.

\bibitem{berthon2007hllc}
Christophe Berthon, Pierre Charrier, and Bruno Dubroca.
\newblock An hllc scheme to solve the {$M_1$} model of radiative transfer in
  two space dimensions.
\newblock {\em Journal of Scientific Computing}, 31(3):347--389, 2007.

\bibitem{CurFial91}
R.~Curto and L.~Fialkow.
\newblock Recursiveness, positivity and truncated moment problems.
\newblock {\em Houston J. Math}, 17(4):603--635, 1991.

\bibitem{dubroca1999theoretical}
Bruno Dubroca and J-L Feugeas.
\newblock Theoretical and numerical study on a moment closure hierarchy for the
  radiative transfer equation.
\newblock {\em Comptes Rendus de l'Academie des Sciences Series I Mathematics},
  329(10):915--920, 1999.

\bibitem{garrett2013comparison}
C~Kristopher Garrett and Cory~D Hauck.
\newblock A comparison of moment closures for linear kinetic transport
  equations: The line source benchmark.
\newblock {\em Transport Theory and Statistical Physics}, 42(6-7):203--235,
  2013.

\bibitem{grad1949kinetic}
Harold Grad.
\newblock On the kinetic theory of rarefied gases.
\newblock {\em Communications on pure and applied mathematics}, 2(4):331--407,
  1949.

\bibitem{levermore1996moment}
C~David Levermore.
\newblock Moment closure hierarchies for kinetic theories.
\newblock {\em Journal of Statistical Physics}, 83(5-6):1021--1065, 1996.

\bibitem{mcclarren2008solutions}
Ryan~G McClarren, James~Paul Holloway, and Thomas~A Brunner.
\newblock On solutions to the {$P_n$} equations for thermal radiative transfer.
\newblock {\em Journal of Computational Physics}, 227(5):2864--2885, 2008.

\bibitem{minerbo1978maximum}
Gerald~N Minerbo.
\newblock Maximum entropy eddington factors.
\newblock {\em Journal of Quantitative Spectroscopy and Radiative Transfer},
  20(6):541--545, 1978.

\bibitem{monreal2008higher}
Philipp Monreal and Martin Frank.
\newblock Higher order minimum entropy approximations in radiative transfer.
\newblock {\em arXiv preprint arXiv:0812.3063}, 2008.

\bibitem{olbrant2012realizability}
Edgar Olbrant, Cory~D Hauck, and Martin Frank.
\newblock A realizability-preserving discontinuous galerkin method for the m1
  model of radiative transfer.
\newblock {\em Journal of Computational Physics}, 231(17):5612--5639, 2012.

\bibitem{toro2009riemann}
Eleuterio~F Toro.
\newblock {\em Riemann solvers and numerical methods for fluid dynamics: a
  practical introduction}.
\newblock Springer Science \& Business Media, 2009.

\bibitem{vikas2013radiation}
V~Vikas, CD~Hauck, ZJ~Wang, and Rodney~O Fox.
\newblock Radiation transport modeling using extended quadrature method of
  moments.
\newblock {\em Journal of Computational Physics}, 246:221--241, 2013.

\end{thebibliography}
